\documentclass[letterpaper, 10 pt, conference]{ieeeconf}  % Comment this line out
\IEEEoverridecommandlockouts                              % This command is only
\usepackage{amsmath}
\usepackage{amssymb} 
\usepackage{algpseudocode}
\usepackage{caption}
\usepackage{subcaption}
\usepackage{graphicx}
\usepackage{url}
\usepackage{color}

\usepackage{tikz}
\usepackage{tkz-graph}
\usetikzlibrary{arrows}
\tikzstyle{block} = [rectangle, draw, 
    text width=12em, text centered, rounded corners, minimum height=2em]
    
\newcommand{\BEAS}{\begin{eqnarray*}}
\newcommand{\EEAS}{\end{eqnarray*}}
\newcommand{\BEA}{\begin{eqnarray}}
\newcommand{\EEA}{\end{eqnarray}}
\newcommand{\BEQ}{\begin{equation}}
\newcommand{\EEQ}{\end{equation}}
\newcommand{\BIT}{\begin{itemize}}
\newcommand{\EIT}{\end{itemize}}
\newcommand{\BNUM}{\begin{enumerate}}
\newcommand{\ENUM}{\end{enumerate}}

\newcommand{\ie}{{\it i.e.}}

\newcommand{\ones}{\mathbf 1}

\newcommand{\diag}[1]{\mathop{\bf diag}\left({#1}\right)}

\newcommand{\trace}[1]{\mathop{\bf Tr}{\left(#1\right)}}
\newcommand{\card}[1]{\mathop{\bf Card}{\left(#1\right)}}
\newcommand{\re}[1]{\mathop{\bf Re}{\left(#1\right)}}
\newcommand{\hinf}{\mathcal{H}_{\infty}}
\newcommand{\range}[1]{\mathop{\bf Range}{\left(#1 \right)}}
\renewcommand{\ker}[1]{\mathop{\bf Ker}{\left(#1\right)}}

\usepackage[amsmath,thmmarks]{ntheorem}
\RequirePackage{amssymb}
\theoremstyle{plain}
\theoremheaderfont{\normalfont\bfseries}
\theorembodyfont{\itshape}
\theoremseparator{:}
\qedsymbol{\ensuremath{_\blacksquare}}

\newtheorem{theorem}{Theorem}

\newtheorem{proposition}[theorem]{Proposition}
\newtheorem{example}{Example}

\newcommand{\vct}[1]{\mathbf{#1}}
\newcommand{\RHinf}{\mathcal{RH}_\infty}
\newcommand{\Hinf}{\mathcal{H}_\infty}
\newcommand{\M}{\mathcal{M}}
\newcommand{\knorm}[1]{\bar{\mu}_k\left({#1}\right)}
\newcommand{\mink}[1]{\underline{\mu}_k\left({#1}\right)}
\newcommand{\ksdpnorm}[1]{\bar{\mu}^{sdp}_k\left({#1}\right)}
\newcommand{\kroundnorm}[1]{\bar{\mu}^{round}_k\left({#1}\right)}

\title{\LARGE \bf
A Convex Approach to Sparse $\Hinf$ Analysis \& Synthesis
}

\author{Seungil You and Nikolai Matni% <-this % stops a space
\thanks{This research was in part supported by NSF NetSE, AFOSR, the Institute
for Collaborative Biotechnologies through grant W911NF-09-0001 from the
U.S. Army Research Office}% <-this % stops a space
\thanks{The authors are with the Control and Dynamical Systems, California Institute of Technology, Pasadena, CA 91125, USA
		{\tt\small \{syou,nmatni\}@caltech.edu}}
}

\begin{document}

\maketitle
\thispagestyle{empty}
\pagestyle{empty}

%%%%%%%%%%%%%%%%%%%%%%%%%%%%%%%%%%%%%%%%%%%%%%%%%%%%%%%%%%%%%%%%%%%%%%%%%%%%%%%%
\begin{abstract}
In this paper, we propose a new robust analysis tool motivated by large-scale systems.  
The $\Hinf$ norm of a system measures its robustness by quantifying the worst-case behavior of a system perturbed by a unit-energy disturbance.  However, the disturbance that induces such worst-case behavior requires perfect coordination among all disturbance channels.  Given that many systems of interest, such as the power grid, the internet and automated vehicle platoons, are large-scale and spatially distributed, such coordination may not be possible, and hence the $\Hinf$ norm, used as a measure of robustness, may be too conservative.  We therefore propose a cardinality constrained variant of the $\Hinf$ norm in which an adversarial disturbance can use only a limited number of channels.  As this problem is inherently combinatorial, we present a semidefinite programming (SDP) relaxation based on the $\ell_1$ norm that yields an upper bound on the cardinality constrained robustness problem.
We further propose a simple rounding heuristic based on the optimal solution of SDP relaxation which provides a lower bound.  Motivated by privacy in large-scale systems, we also extend these relaxations to computing the minimum gain of a system subject to a limited number of inputs.  Finally, we also present a SDP based optimal controller synthesis method for minimizing the SDP relaxation of our novel robustness measure.  The effectiveness of our semidefinite relaxation is demonstrated through numerical examples. 

\end{abstract}

%%%%%%%%%%%%%%%%%%%%%%%%%%%%%%%%%%%%%%%%%%%%%%%%%%%%%%%%%%%%%%%%%%%%%%%%%%%%%%%%
\section{Introduction}
Structure, and in particular sparsity, has proven to be a powerful tool in the analysis and design of large-scale control systems. Lyapunov analysis \cite{MP14_chordal}, distributed performance certification \cite{MLP14_admm}, distributed optimal controller synthesis \cite{RL06_QI} and controller architecture design \cite{matni2014regularization} all rely on and exploit structural properties of the underlying system to solve seemingly intractable problems in a computationally efficient manner.  In contrast, in the context of robust control, adding additional structure to system uncertainty has traditionally made analysis and synthesis more difficult.  For instance, linear matrix inequality (LMI) based necessary and sufficient conditions for the robust stability of a system subject to an unstructured delta block can be derived, but no such results exist if we restrict ourselves to highly structured delta blocks \cite{Dullerud:2010tc}. 

In this paper, we ask the following question, which we later interpret in terms of robustness to structured disturbances: given a large scale system with $p$ input channels, what $k\ll p$ input channels should be used to maximally (minimally) perturb the system using an unit energy input. We show that the solution can be obtained by suitably modifying the power semi-norm based definition of the $\Hinf$ norm of a system to incorporate a cardinality constraint on the input; we therefore call the resulting performance metric the $k$-sparse $\Hinf$ norm of the system.

We argue that questions pertaining to the maximal and minimal gains of a system restricted to a sparse subset of inputs arise naturally in the context of distributed system robustness analysis, consensus robustness analysis, privacy and system security.  We further show that the resulting optimization problems are in fact a generalization of the maximal and minimal sparse eigenvalue problems, objects of central importance in certifying the performance of compressed sensing matrices \cite{candes2008restricted} and in sparse PCA \cite{d2007direct}. We also show touch upon how these restricted gains relate to analogous conditions developed in the Regularization for Design (RFD) \cite{matni2014regularization} framework that guarantee the recovery of optimal controller architectures.

Of course, the resulting optimization problems are combinatorial in nature, and are easily seen to be computationally difficult in general. Leveraging a novel \emph{primal} formulation of the KYP lemma \cite{you2014h}, we propose a semidefinite relaxation (akin to that proposed in \cite{d2007direct}) for computing lower/upper bounds on the resulting minimal/maximal restricted gains of the system, and a simple rounding heuristic to obtain corresponding upper/lower bounds.  We further derive the dual of the resulting semidefinite program (SDP) and show that it has similar structure to the traditional KYP LMI test, allowing for standard semidefinite programming based controller synthesis methods to be applied.

The paper is organized as follows: in Section II, we formally introduce the $k$-sparse $\Hinf$ norm and the analogous $k$-sparse minimal gain of a system, and elucidate on several engineering applications.  We also make connections to compressed sensing, restricted isometry constants and sparse PCA, as well as RFD.  In Section III we present both a semidefinite relaxation and a rounding heuristic to obtain lower and upper bounds on the $k$-sparse $\Hinf$ norm of a system.  The dual to our semidefinite relaxation is derived in Section IV, and we show how it can be used to synthesize a centralized controller that minimizes the relaxed $k$-sparse $\Hinf$ norm of the system.  We present several numerical examples in Section V, and end with a summary and discussion of future directions in Section VI.

\subsection{Notation}
We use $\RHinf$ to denote the space of stable real-rational proper transfer matrices.
We use lower case Latin letters $x$ to denote vectors, bold lower case Latin letters $\vct x$ to denote signals, upper case Latin letters $X$ to denote matrices and upper case calligraphic letters $\mathcal{X}$ to denote elements of $\RHinf$. 

We recall the definition of the power semi-norm, $\|\mathbf{x}\|^2_P := \underset{N \rightarrow \infty}{\lim} \frac{1}{N}\sum_{k=0}^{N-1} x_k^*x_k$.
For a matrix $X$, we denote its conjugate transpose by $X^*$, transpose by $X^{\top}$, the projection of $X$ onto its diagonal elements by $\diag{X}$, and the range space and the null space of $X$ by $\range{X}$ and $\ker{X}$ respectively.
In addition, $|X|$ denotes the element-wise absolute value of $X$, and the one vector $\ones$ is a vector whose entires are all one.
The generalized inequality $X \succeq 0$ means that $X$ is positive semidefinite, and $X \succ 0$ means that $X$ is positive definite.
\section{$k$-sparse $\mathcal{H}_{\infty}$ analysis}
%These are presented without proof at the end of the section.
We consider a discrete time linear time invariant system\footnote{ Although we present our analysis for discrete-time systems, analogous arguments and results hold for continuous time systems.}
\begin{equation}
\mathcal{M}(z) = C(z I - A)^{-1}B + D \in \RHinf.
\label{eq:system}
\end{equation}

 Recall that the $\Hinf$ norm of $\M$ can be computed as the worst case gain in the output of the system induced by a disturbance of unit power semi-norm \cite{doyle1992feedback,you2013lagrangian,you2014h}:
\begin{equation}
\begin{aligned}
\|\mathcal{M}\|_{\infty}^2:=& \underset{\mathbf{w}, \mathbf{x}} {\text{maximize}}
& & \|C\mathbf{x}+D\mathbf{w}\|_P^2\\
&\text{subject to}
&& x_{k+1} = Ax_k + Bw_k\\
&&& x_0 = 0\\
&&& \|\mathbf{w}\|_P^2 \leq 1.
\end{aligned}
\label{eq:hinfinity}
\end{equation}
The $\Hinf$ norm measures the worst-case behavior of the system subject to power semi-norm bounded disturbances, and it has well known implications on the robust stability of the system with uncertain blocks \cite{Dullerud:2010tc}, as well as many practical interpretations \cite{doyle1992feedback}.

One such interpretation is that an attacker seeks to maximize their disruption of the system using the disturbance $\vct w$ -- in this case, the optimal disturbance $\vct w^\star$ to optimization problem \eqref{eq:hinfinity} is precisely a disturbance that maximizes the attacker's impact on the system.  Taking an opposite perspective, from the viewpoint of a system designer, the maximizing disturbance denotes a weak point of the system that may need to be addressed.

A seemingly innocuous assumption in the above analysis is that the attacker can simultaneously coordinate all of the disturbance channels: although reasonable in a centralized setting, this assumption may prove to be quite conservative when $\M$ is a distributed system.
In particular, if there are many possible disturbances ($B$ has many columns), and these disturbances enter through channels that are physically separated, it may be overly conservative to consider the response of the system to a centralized attack. 
In order to alleviate this conservativeness, we propose a \emph{cardinality constrained} variation of optimization problem \eqref{eq:hinfinity}, in which we assume that at most $k$ disturbance channels can have non-zero power semi-norms\footnote{We define the cardinality of a power signal, $\card{\mathbf{w}}$ as the number of indices $i$ such that  $\|\mathbf{w}_i\|_P > 0.$
Notice that because we use the power semi-norm in this definition, all signals with finite $\ell_2$ norm have a cardinality of $0$, as their power semi-norm is $0$.}:
\begin{equation}
\begin{aligned}
\{\knorm{\M}\}^2 := ~
& \underset{\mathbf{w}, \mathbf{x}} {\text{maximize}}
& & \|C\mathbf{x}+D\mathbf{w}\|_P^2\\
&\text{subject to}
&& x_{k+1} = Ax_k + Bw_k\\
&&& x_0 = 0\\
&&& \|\mathbf{w}\|_P^2 \leq 1\\
&&& \card{\mathbf{w}} \leq k.
\end{aligned}
\label{eq:k-sparse-hinfinity}
\end{equation}
We refer to $\knorm{\M}$ as the $k$-sparse $\Hinf$ norm of system $\M$. 

It should be clear that $\knorm{\M} \leq \|\mathcal{M}\|_{\infty}$ for all $k$, 
but the size of the difference between those two quantities is unclear.
If the gap is small, then the additional effort needed to accommodate the cardinality constraint on the disturbance may not be justified.  
Before elaborating on other interpretations of the $k$-sparse $\Hinf$ norm of a system, we show that the gap between $\knorm{\cdot}$ and $\|\cdot\|_\infty$ can be made arbitrarily large for a fixed $k$ by letting the state dimension of the underlying system tend to infinity.

\begin{example}
Consider a system $\M$, as in \eqref{eq:system}, described by state-space parameters $(A,B,C,D)$, where $A =  0.99 \frac{1}{n} \mathbf{1}\mathbf{1}^{\top} + 0.1 (I_n -  \frac{1}{n} \mathbf{1}\mathbf{1}^{\top})\in \mathbb{R}^{n \times n}$, $B = I_n$, $C = I_n$, $D = 0_{n,n}$.
%Furthermore $A$ is symmetric, and $\lambda_1(A) = 0.99$, and the corresponding eigenvector $v_1 = \mathbf{1}$, and $\lambda(A) = 0.1$ for $i \geq 2$.
Due to the special structure of the state-space parameters, optimization problems \eqref{eq:hinfinity} and  \eqref{eq:k-sparse-hinfinity} can be solved analytically, and $\frac{\knorm{\M}}{\|\M\|_\infty} = O\left(\sqrt{\frac{k}{n}}\right)$.
Thus for a fixed $k$ the gap between $\knorm{\M}$ and $\|\M\|_\infty$ can be made arbitrarily large by letting $n\to\infty$.  
Figure \ref{fig:gap} shows  $\frac{\knorm{\M}}{\|\M\|_\infty}$ for $k=5$ and $n=5,\dots,30$.
\label{ex:motivating}
\end{example}

\begin{figure}[ht!]
\includegraphics[width=0.5\textwidth]{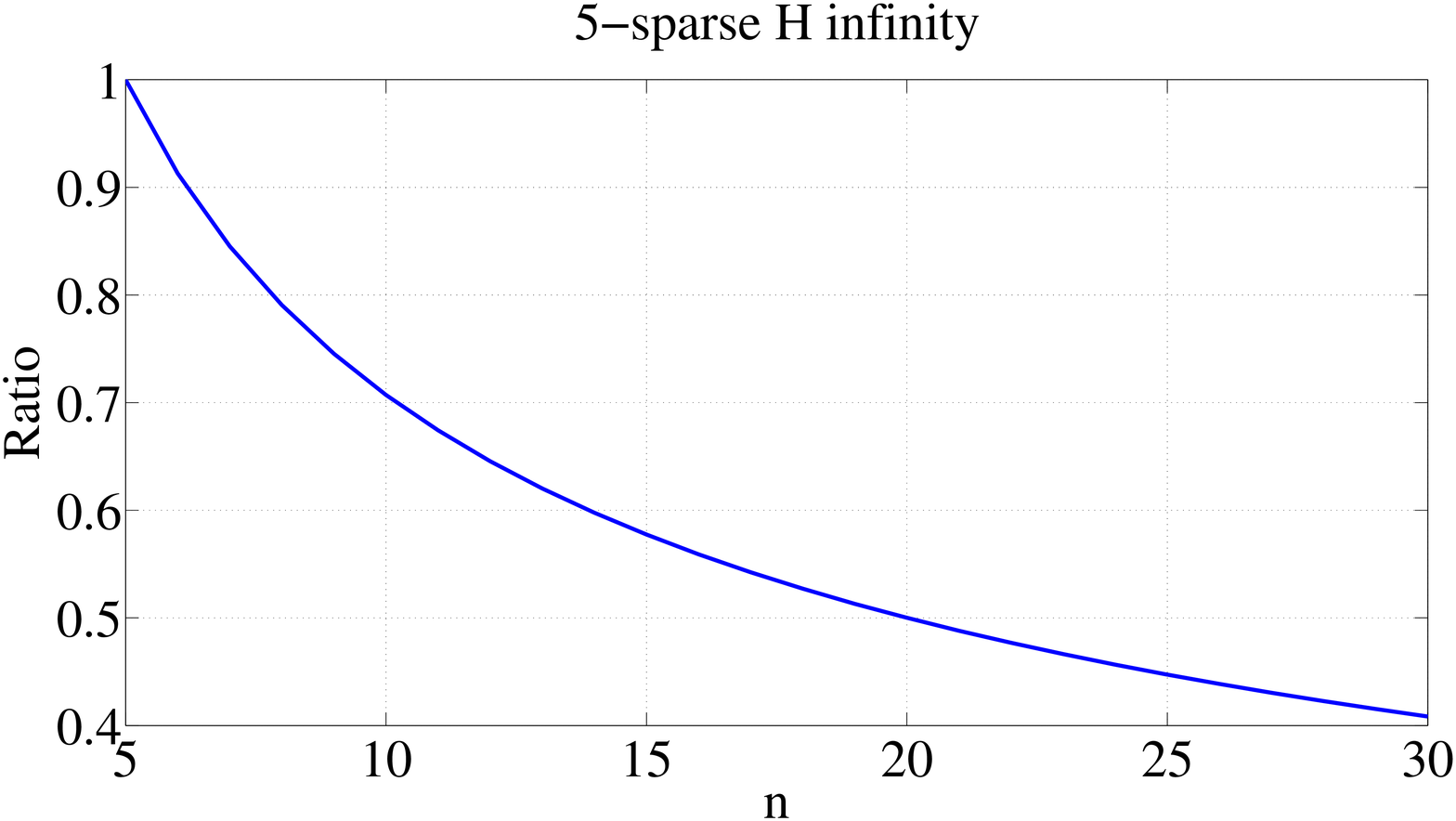}
\caption{The ratio, $\bar{\mu}_5(\M)/\|\M\|_{\infty}$ for $n=5, \cdots, 30$.}
\label{fig:gap}
\end{figure}

Example \ref{ex:motivating} shows the existence of systems for which standard $\Hinf$ analysis can be overly conservative if indeed only at most $k$ disturbances can be expected to coordinate their efforts to disrupt the system.  We now outline two concrete engineering applications in which such an analysis may arise.
%
%AlthoughIn a networked system point of view, this means that all disturbance channel coordinates with each other to maximize (or minimize) its impact on the system which may not be a realistic scenario.

\textbf{Robustness analysis for distributed system:} Quantifying the robustness of a distributed system, such as the power grid, allows the system designer to plan for and mitigate the worst case effects of un-modeled dynamics and disturbances.  
The need for robustness is increasingly important in the context of the power grid as it becomes more reliant on intermittent distributed energy resources, such as renewables.  
However, as mentioned, $\hinf$ analysis assumes that all such distributed energy resources coordinate with each other to destabilize the power network, which may be overly conservative and lead to loss of efficiency.  
Rather, we propose using the $k$-sparse $\hinf$ norm of the system to identify and quantify vulnerabilities of the system to potentially more realistic disturbances.
%In addition, based on our analysis, the system operator may avoid to install an uncertain system on such weak points, which is a complimentary tool for the regularization for the design \cite{matni2014regularization}

\textbf{Robustness analysis for consensus network:}
The well studied problem of consensus (or synchronization) \cite{OlfatiSaber:cm,Boyd:bu, jadbabaie2003coordination} is one in which a set of agents seek to converge to a common value using simple local averaging rules. 
When these local rules are linear and time invariant, the consensus protocol can be modeled as an LTI system.
In this case, a system dynamics $A$ satisfies the following properties \cite{Xiao:2004gg}: $A\mathbf{1} = \mathbf{1}$, $A^{\top}\mathbf{1} = \mathbf{1}$, and $\rho(A - \frac{1}{n} \mathbf{1}\mathbf{1}^{\top}) < 1$, where $n$ is the number of nodes in the network. 

Although typically considered in a disturbance free setting, it is also natural to ask how much local disturbances applied to individual agents can affect the system's ability to reach consensus. 
Concretely, assume that each agent can be corrupted by a separate disturbance, i.e., that $B = I_n$, and  we measure the effect of the disturbances on the deviation of each state $x^j_k$ from the consensus value, as encoded by $z^j_k = x^j_k - \frac{1}{n}  \sum_i x^i_k$, such that $C = I_n - \frac{1}{n}\mathbf{1}\mathbf{1}^{\top}$, and $D = 0$.
Note that the marginally stable mode of $A$ is unobservable with respect to the measured output defined by $C$, and the system has a finite $\Hinf$ norm and $k$-sparse $\Hinf$ norm.

Whereas the $\Hinf$ norm of the resulting system measures the effects of a worst-case attack on all agents, the $k$-sparse $\Hinf$ norm measures the effects of worst-case attack on only $k$ agents.
From an attacker's perspective, this may result in a more realistically implementable strategy, and from a system designer's perspective, this provides valuable information as to which agents should be most closely monitored and protected from attack.

\subsection{The $k$-sparse minimal gain of a system}
We can also define the minimal $k$-sparse gain of system $\M$, which we denote by $\mink{\M}$ as
\begin{equation}
\begin{aligned}
\{\mink{\M}\}^2 := ~
& \underset{\mathbf{w}, \mathbf{x}} {\text{minimize}}
& & \|C\mathbf{x}+D\mathbf{w}\|_P^2\\
&\text{subject to}
&& x_{k+1} = Ax_k + Bw_k\\
&&& x_0 = 0\\
&&& \|\mathbf{w}\|_P^2 \geq 1\\
&&& \card{\mathbf{w}} \leq k.
\end{aligned}
\label{eq:minimization}
\end{equation}

\textbf{Privacy:} An immediate interpretation of this optimization problem is in terms of \emph{privacy}.  Suppose that a publicly available variable is defined by $z_k = Cx_k$, and that a user wishes to transfer at least $\gamma$ units of power to $y_k = Gx_k + Hw_k$ while minimizing their effect on the public variable.  The optimal action for the user to take can be determined by solving optimization problem \eqref{eq:minimization} with the added constraint
\begin{equation}
 \|G\mathbf{x} + H \mathbf{w}\|_P^2 \geq \gamma^2
 \label{eq:privacy}
\end{equation}

\textbf{System security:} One can also view the user in the above scenario as an attacker, and the publicly available variable as a system monitor: in this case, the optimal input $\vct w_\star$ corresponds to the least detectable input that still disrupts the output $\vct y$ by $\gamma$ units of power.  
Allowing for sparse optimal inputs $\vct w_\star$ makes for more realistically implementable actions by either a user or an attacker.

% The following is from Nik and integrated by SYOU
\subsection{Connections to the Restricted Isometry Property and Regularization for Design}
Our problem formulation seeks the minimal and maximal gains of a linear operator restricted to $k$-sparse subspaces.
When the linear operator is a static matrix $D$, instead of a dynamical system ($A = B = C = 0$), then the cardinality constrained optimization problems \eqref{eq:k-sparse-hinfinity} and \eqref{eq:minimization} compute precisely the maximal and minimal \emph{restricted eigenvalues}  \cite{bickel2009simultaneous} of the matrix $D^\top D$, that is the maximal and minimal gains of $D$ restricted to sparse subspaces. They are also closely linked to the Restricted Isometry Property (RIP) constant of the matrix, which can be used to state conditions for the recovery of sparse vectors \cite{candes2008restricted} via convex optimization, and can be used to perform sparse principal component analysis (sPCA) \cite{d2007direct}.  
We can therefore view optimization problem \eqref{eq:k-sparse-hinfinity} as a tool for bounding the restricted eigenvalues of an infinite dimensional LTI operator acting on signals in $\ell_2$.  

Moreover, the $k$-sparse $\Hinf$ norm and the $k$-sparse minimal gain of a system also have natural connections to the Regularization for Design (RFD) framework developed in  \cite{matni2014regularization}.  In the RFD framework, atomic norms \cite{chandrasekaran2012convex} are added as convex penalties to traditional model matching problems in order to design architecturally simple controllers.  Further, control theoretic analogs to the recovery conditions found in the structured inference literature are stated in terms of restricted gains that are closely related to the $k$-sparse $\Hinf$ norm and $k$-sparse minimal gain of a system -- we are currently actively exploring the application of the computational methods developed in this paper to computing bounds on these restricted gains.

\section{SDP relaxation of $k$-sparse $\mathcal{H}_{\infty}$ analysis}
%\subfile{../chapters/review_of_h_infinity}
As posed, optimization problems \eqref{eq:k-sparse-hinfinity}  and \eqref{eq:minimization} are intractable:  the optimization variables are infinite dimensional, and the cardinality constraint introduces a combinatorial aspect to the problem.  In order to develop a computationally tractable framework, we propose an SDP based convex relaxation of the $k$-sparse $\hinf$ norm \eqref{eq:k-sparse-hinfinity} and the $k$-sparse minimal gain \eqref{eq:minimization}.
We begin by reviewing recent results on traditional $\Hinf$ analysis \cite{you2014h,gattami2013simple}.

\subsection{Review of $\hinf$ analysis}
From previous work \cite{you2014h,gattami2013simple}, we know that using a matrix $$V := \underset{n \rightarrow \infty}{\lim} \frac{1}{N}\sum_{k=0}^{N-1} \begin{bmatrix} {x}_{k} \\ {w}_k\end{bmatrix}\begin{bmatrix} {x}_{k} \\ {w}_k\end{bmatrix}^* \succeq 0,$$ the optimization \eqref{eq:hinfinity} can be transformed to the following equivalent finite dimensional semidefinite program:

\begin{equation}
\begin{aligned}
& \underset{V \succeq 0} {\text{maximize}}
& & \trace{\begin{bmatrix} C^*C & C^*D\\D^*C & D^*D\end{bmatrix}V}\\
&\text{subject to}
&& \begin{bmatrix} I & 0 \end{bmatrix} V \begin{bmatrix} I\\ 0\end{bmatrix} = \begin{bmatrix} A & B \end{bmatrix} V \begin{bmatrix} A^*\\B^* \end{bmatrix}\\
&&& \sum_{i=n+1}^{n+m} V_{ii} \leq 1,
\end{aligned}
\label{eq:hinfinity-sdp}
\end{equation}
where $n$ is the dimension of the state $\mathbf{x}$, and $m$ is the dimension of the disturbance $\mathbf{w}$.
The key idea of the proof is to construct the sinusoid $\mathbf{w}$ that achieves the $\mathcal{H}_{\infty}$ norm using a rank one solution of the semidefinite program \eqref{eq:hinfinity-sdp}.

In the construction of $\mathbf{w}$, there is no prior structure imposed on $\mathbf{w}$.
This means that, in general, all $m$ disturbance channel are active and must coordinate amongst themselves.

\subsection{SDP relaxation of $k$-sparse $\mathcal{H}_{\infty}$ analysis}
Building on the result of the previous section, we propose and analyze a semidefinite relaxation of the optimization problem \eqref{eq:k-sparse-hinfinity} that can be used to compute an \emph{upper bound} to the $k$-sparse $\hinf$ norm of a system.  The relaxation to the $k$-sparse minimal gain of a system \eqref{eq:minimization} is analogous, and stated without proof.

To begin with, let us use the matrix $$V := \underset{n \rightarrow \infty}{\lim} \frac{1}{N}\sum_{k=0}^{N-1} \begin{bmatrix} {x}_{k} \\ {w}_k\end{bmatrix}\begin{bmatrix} {x}_{k} \\ {w}_k\end{bmatrix}^* \succeq 0,$$ as in $\mathcal{H}_{\infty}$ analysis.
For notational convenience, we partition $V = \begin{bmatrix} X & R\\R^* & W\end{bmatrix}$ where $X \in \mathbb{C}^{n \times n}$, and $W \in \mathbb{C}^{m\times m}$.

Now we have the following proposition.

\begin{proposition}
${\card{\mathbf{w}}} \leq k$
 if and only if  $\card{\diag{W}} \leq k$.
\end{proposition}
\begin{proof}
From the definition, $W_{ii} = \|\mathbf{w}_i\|_P^2$, where $W_{ii}$ is $(i,i)$ entry of $W$ and $\mathbf{w}_i$ is the $i$th component of the vector-valued signal $\mathbf{w}$.
Therefore ${\card{\mathbf{w}}} = \card{\diag{W}}$.
\end{proof}
%This shows a clear advantage using the approach based on \cite{you2014h,gattami2013simple} that is the matrix $V$ contains an information on the disturbance $\mathbf{w}$.
%On the other hand, Kalman--Yakubovich--Popov (KYP) lemma \cite{Rantzer:2011wn} based approach does not contain an information on $\mathbf{w}$, so it may not be possible to impose this type of structural requirement on the disturbance.

By applying the same procedure from \cite{you2014h} used to derive the SDP used for $\mathcal{H}_{\infty}$ analysis, we obtain the following optimization problem, which provides an upper bound of \eqref{eq:k-sparse-hinfinity}.

\begin{equation}
\begin{aligned}
& \underset{X, R, W} {\text{maximize}}
& & \trace{\begin{bmatrix} C^*C & C^*D\\D^*C & D^*D\end{bmatrix}\begin{bmatrix} X & R\\ R^* & W\end{bmatrix}}\\
&\text{subject to}
&& X = \begin{bmatrix} A & B \end{bmatrix} \begin{bmatrix} X & R\\ R^* & W\end{bmatrix}\begin{bmatrix}A^*\\B^*\end{bmatrix}\\
&&& \trace{W} \leq 1\\
&&& \card{\diag{W}} \leq k\\
&&& \begin{bmatrix} X & R\\ R^* & W\end{bmatrix} \succeq 0.
\end{aligned}
\label{eq:k-sparse-hinfinity-sdp}
\end{equation}

For the standard $\Hinf$ problem, this SDP relaxation is tight: the proof consists of constructing a disturbance $\mathbf{w}$ that achieves the optimal value of the SDP.
Similarly, once a solution to optimization problem \eqref{eq:k-sparse-hinfinity-sdp} is obtained, we can consider a system with disturbance inputs specified by the support of the optimal disturbance, and thus apply the methods of  \cite{you2014h}.
Thus, the cardinality constrained SDP \eqref{eq:k-sparse-hinfinity-sdp}  is in fact equivalent to $k$-sparse $\hinf$ optimization \eqref{eq:k-sparse-hinfinity}.

%By restricting the support of the disturbance, we can easily construct such $\mathbf{w}$ whose cardinality is up to $k$
%By restricting the support of the disturbance, we can easily show that this relaxation is exac
%In this case, such a construction is unknown, so this will be an upper bound in general. \textcolor{blue}{Is this true?  For each fixed sparsity pattern assigned to $\mathbf w$, the problem reduces to a standard $\Hinf$ problem, so shouldn't the relaxation be exact?}
In applying the techniques from \cite{you2014h}, we have reduced the optimization problem to a finite dimensional semidefinite program with an added cardinality constraint $\card{\diag{W}} \leq k$.  In order to circumvent the intractability of this constraint, we propose using an $\ell_1$ relaxation \cite{tibshirani1996regression}.  This approach is inspired by \cite{d2007direct}, in which the authors consider the $\ell_1$ relaxation of an analogous cardinality constraint to obtain a semidefinite relaxation of the  sparse PCA problem, in which one seeks the leading sparse singular vector of a matrix (as mentioned previously, this is closely related to the RIP constant of a matrix and to analogous quantities in RFD).
%An extension of this idea in the sparse PCA setting where one would like to find a sparse singular vector of a matrix, the authors consider a $l_1$ relaxation of the cardinality constraint to obtain a SDP relaxation of the sparse PCA problem, and show that the true optimal solution can be obtained in certain cases.
In order to adapt this idea to our problem formulation, we need the following observation.

\begin{proposition}
Consider $W \in \mathbb{C}^{n \times n}$ such that $W \succeq 0$, $\trace{W} \leq 1$.
Then, $\mathbf{1}^T|W|\mathbf{1} \leq n$.
\label{prop:l1-bound}
\end{proposition}
\begin{proof}
Consider a Hermitian matrix $H$ where
\BEAS
H_{ij} &=& 
\begin{cases}
1 & \text{if $i = j$} \\
 e^{\mathbf{i}\theta_{ij}} & \text{if $i \neq j$},
 \end{cases}
\EEAS
for some $\theta_{ij}$.
If we construct ${H}$ such that ${H}_{ij} = e^{\mathbf{i} \angle W_{ij}}$, then 
$\mathbf{1}^T|W|\mathbf{1} = \trace{{H}^*W}$.
This shows that $\mathbf{1}^T|W|\mathbf{1} \leq \sup_H \trace{H^*W}$,
and from the Von Neumann's trace inequality \cite{mirsky1975trace}, we have
\BEAS
\trace{H^*W} \leq \sum_{i} \sigma_i(W)\sigma_i(H),
\EEAS
where $\sigma_i$ is the $i$th singular value of the matrix.
Furthermore, by definition of $H$ we have $\sigma_1(H) \leq \sum_i \sigma_i(H) = \trace{H} = n$. 
Therefore,
\BEAS
&&\trace{H^*W}
\leq \sum_{i} \sigma_i(W)\sigma_i(H)
\leq \sigma_1(H) \sum_{i} \sigma_i(W)\\
&\leq& n \trace{W} \leq n,
\EEAS
and $\mathbf{1}^T|W|\mathbf{1} \leq \sup_H \trace{H^*W} \leq n$.
Notice that this upper bound is achieved by $W = \frac{1}{n}\mathbf{1}\mathbf{1}^T$, which shows the inequality is tight.
\end{proof}
We can now connect the $\ell_1$ norm bound to the cardinality constraint of optimization problem \eqref{eq:k-sparse-hinfinity}.
\begin{proposition}
Consider a positive semidefinite matrix $W$ with $\trace{W} \leq 1$ and $\card{\diag{W}} \leq k$.
Then, $\mathbf{1}^T|W|\mathbf{1} \leq k$.
\label{prop:general-l1-bound}
\end{proposition}
\begin{proof}
Without loss of generality, we can assume that $W_{11}, \cdots, W_{ii}$ for $1 \leq i \leq k$ are not zero,  where $W_{ii}$ is $(i,i)$th entry of $W$.
Then from the Schur complement, we can easily check that $W$ should have the form 
\BEAS
W =  \left[\begin{array}{ccc} \tilde{W} & \vline & 0\\ \hline 0 & \vline & 0\end{array}\right],
\EEAS
where $\tilde{W}$ is a $i \times i$ Hermitian matrix.
Therefore, from  Proposition \ref{prop:l1-bound}, $\mathbf{1}^T|W|\mathbf{1} = \mathbf{1}^T|\tilde{W}|\mathbf{1} \leq i \leq k$, which concludes the proof.
\end{proof}

In the cardinality constrained problem \eqref{eq:k-sparse-hinfinity-sdp}, the $W$ matrix satisfies the requirements of Proposition \ref{prop:general-l1-bound}. 
This shows that if we replace the cardinality constraint on $W$ in optimization problem \eqref{eq:k-sparse-hinfinity-sdp} by a suitable $\ell_1$ norm bound, then we have a larger feasible set.
Although this procedure provides an upper bound to \eqref{eq:k-sparse-hinfinity-sdp}, the resulting optimization becomes a semidefinite program, which can be solved efficiently \cite{boyd2004convex}.
Therefore, we propose the following $\ell_1$ based relaxation of \eqref{eq:k-sparse-hinfinity-sdp}, which is the main optimization problem in this paper.

\begin{align}
\{\ksdpnorm{\M}\}^2 := ~
& \underset{X, R, W} {\text{max}}
& & \trace{\begin{bmatrix} C^*C & C^*D\\D^*C & D^*D\end{bmatrix}\begin{bmatrix} X & R\\ R^* & W\end{bmatrix}} \nonumber\\
&\text{s.t.}
&& X = \begin{bmatrix} A & B \end{bmatrix} \begin{bmatrix} X & R\\ R^* & W\end{bmatrix}\begin{bmatrix}A^*\\B^*\end{bmatrix}\nonumber\\
&&& \trace{W} \leq 1 \label{eq:sparse-hinfinity-sdp}\\
&&& \mathbf{1}^T|W|\mathbf{1} \leq k\nonumber\\
&&& \begin{bmatrix} X & R\\ R^* & W\end{bmatrix} \succeq 0.\nonumber
\end{align}
In addition, it should be obvious that $\knorm{\M} \leq \ksdpnorm{\M}$.
%Notice that this is an SDP and can be solved very efficiently. 
A careful remark is that for a complex matrix, $|W|$ should be treated as an SOCP, not an LP  \cite{kim2007interior}.
%Since $V_{ii}$ is always real and nonnegative, we do not need an absolute value of $V_{ii}$.
%By increasing $\lambda \geq 0$, we can encourage the sparsity.

Although we omit the details, a similar argument for continuous time systems yields
\begin{equation}
\begin{aligned}
& \underset{X,R,W} {\text{maximize}}
& & \trace{\begin{bmatrix} C^*C & C^*D\\D^*C & D^*D\end{bmatrix}\begin{bmatrix} X & R\\ R^* & W\end{bmatrix}}\\
&\text{subject to}
&& XA^* + AX + R^*B^* + BR = 0\\
&&&\trace{W} \leq 1\\
&&&  \mathbf{1}^T|W|\mathbf{1} \leq k\\
&&& \begin{bmatrix} X & R\\ R^* & W\end{bmatrix} \succeq 0.
\end{aligned}
\label{eq:cont-sparse-hinfinity-sdp}
\end{equation}

\subsection{Extension to $k$-sparse minimal gain}
In the previous section, we introduced a $k$-sparse minimal gain.
A similar approach can be used to obtain the following SDP relaxation of \eqref{eq:minimization}.

\begin{align}
\{\underline{\mu}^{sdp}_k(\M)\}^2 := ~
& \underset{X, R, W} {\text{min}}
& & \trace{\begin{bmatrix} C^*C & C^*D\\D^*C & D^*D\end{bmatrix}\begin{bmatrix} X & R\\ R^* & W\end{bmatrix}} \nonumber\\
&\text{s.t.}
&& X = \begin{bmatrix} A & B \end{bmatrix} \begin{bmatrix} X & R\\ R^* & W\end{bmatrix}\begin{bmatrix}A^*\\B^*\end{bmatrix}\nonumber\\
&&& \trace{W} \geq 1 \nonumber\\
&&& \mathbf{1}^T|W|\mathbf{1} \leq k\nonumber\\
&&& \begin{bmatrix} X & R\\ R^* & W\end{bmatrix} \succeq 0\nonumber
\end{align}
Similarly, for the continuous time case, we have
\begin{equation}
\begin{aligned}
& \underset{X,R,W} {\text{minimize}}
& & \trace{\begin{bmatrix} C^*C & C^*D\\D^*C & D^*D\end{bmatrix}\begin{bmatrix} X & R\\ R^* & W\end{bmatrix}}\\
&\text{subject to}
&& XA^* + AX + R^*B^* + BR = 0\\
&&&\trace{W} \geq 1\\
&&&  \mathbf{1}^T|W|\mathbf{1} \leq k\\
&&& \begin{bmatrix} X & R\\ R^* & W\end{bmatrix} \succeq 0.
\end{aligned}
\end{equation}
%
%Although the above problem is well-defined and has very interesting implications in the CPS security analysis, but since this is another research direction, we would not further investigate this topic in here.

\subsection{Rounding heuristic for solution refinement}
Let $W^{\star}$ be the optimal solution to optimization problem \eqref{eq:sparse-hinfinity-sdp}.
Since this matrix contains information about the worst-case disturbance, we can extract candidate worst case disturbance channels, and use those to obtain a corresponding lower bound to the value of optimization problem \eqref{eq:k-sparse-hinfinity}.
The approach is simple: identify the top $k$ entries of $\diag{W}$, say $\{W_{i_1i_1}, W_{i_2i_2},\cdots,W_{i_ki_k}\}$, and then restrict $B$ and $D$ to the column space corresponding to these disturbance channels.
We can then compute the traditional $\mathcal{H}_{\infty}$ norm of the system defined by these restricted $B$ and $D$ matrices using classical methods.  As mentioned, as we are choosing specific disturbance channels, this procedure yields a lower bound of the $k$-sparse $\Hinf$ norm \eqref{eq:k-sparse-hinfinity} of a system.
The procedure can thus be summarized as follows:

\textbf{Rounding heuristic:}
\begin{enumerate}
\item Solve \eqref{eq:sparse-hinfinity-sdp} to obtain $W^{\star}$.
\item Find the indices $\{i_1,\cdots,i_k\}$ such that $W^{\star}_{i_1i_1} \geq \cdots \geq W^{\star}_{i_ki_k} \geq \cdots \geq W_{i_ni_n}$.
\item Construct $E := \begin{bmatrix} e_{i_1} & \cdots & e_{i_k}\end{bmatrix} \in \mathbb{R}^{m \times k}$ using a standard basis $\{e_i\} \in \mathbb{R}^m$.
\item Let $\tilde{B} := BE$, $\tilde{D} = DE$, and obtain $\kroundnorm{\M} := \|\tilde{B}(e^{\mathbf{i}\theta}I - A)^{-1}C + \tilde{D}\|_{\infty}$.
\end{enumerate}
Notice that  step 3 chooses $i_1,\cdots i_k$ to be the active disturbance channels.
From this rounding procedure we obtain the inequality
\BEAS
\kroundnorm{\M} \leq \knorm{\M} \leq \ksdpnorm{\M}
\EEAS
Therefore, if the gap between $\kroundnorm{\M}$ and $\ksdpnorm{\M}$ is not large, then $\kroundnorm{\M}$ effectively solves the $k$-sparse $\hinf$ problem and returns a candidate set of worst case disturbance channels.
Notice that this heuristic can also be applied to the continuous time case and the minimal gain computation, but we omit these details.

\section{Dual problem \& Controller Synthesis}
As optimization problem \eqref{eq:sparse-hinfinity-sdp} is an SDP, it is natural to consider its Lagrangian dual problem.
To do this, let us begin with the following observation.
\begin{proposition}
For $w \geq 0$, $\lambda \in \mathbb{C}$, 
\BEAS
\sup_{x \in \mathbb{C}} \{-w|x| + \re{\lambda x}\}
=
\begin{cases}
0 &\text{if $|\lambda| \leq w$}\\
+\infty &\text{otherwise}
\end{cases}.
\EEAS
\label{prop:l1-dual}
\end{proposition}
\begin{proof}
Suppose $|\lambda| > w$. Let $x = \alpha \lambda^*$. Then 
\BEAS
-w|x| + \re{\lambda x} = \alpha|\lambda|(|\lambda|-w).
\EEAS
By taking $\alpha \rightarrow \infty$, we obtain the result.

Suppose $|\lambda| \leq w$. From Cauchy-Schwartz inequality,
\BEAS
-w|x| +  \re{\lambda x} \leq -w|x| + |\lambda||x| \leq (|\lambda|-w)|x| \leq 0,
\EEAS
for all $x \in \mathbb{C}$. Since the upper bound is achieved by $x = 0$, we can conclude the proof.
\end{proof}

With this technical tool in hand, we may proceed to derive the dual to optimization problem \eqref{eq:sparse-hinfinity-sdp}.
First, we form the Lagrangian function in terms of $V = \begin{bmatrix} X & R\\R^* &W\end{bmatrix}$.
\BEAS
&&L(V, P, Q, \lambda, t) := \\
&& \trace{QV} + \trace{\begin{bmatrix} C^*C & C^*D\\D^*C & D^*D\end{bmatrix}V}\\
&&+ \trace{P\left(\begin{bmatrix} A & B \end{bmatrix} V \begin{bmatrix} A^* \\ B^* \end{bmatrix} - \begin{bmatrix} I & 0 \end{bmatrix} V \begin{bmatrix} I \\ 0 \end{bmatrix}\right)}\\
&&+ \lambda\left(1-\trace{\begin{bmatrix} 0 & 0\\0 & I\end{bmatrix}V}\right)\\
&&+ t\left(k - \trace{\begin{bmatrix} 0 & 0\\ 0 & \mathbf{1}\mathbf{1}^T\end{bmatrix} |V|}\right),
\EEAS
where $P = P^*$, $Q \succeq 0$, $\lambda \geq 0$, $t \geq 0$.

Using cyclic property of the trace operator and from Proposition \ref{prop:l1-dual}, we can obtain the dual function $d(Q, P,\lambda,t) := \sup_{V = V^*} L(V,P,\lambda,t)$ which becomes $\lambda + k \cdot t$ when,
\begin{align}
&\left|Q 
+ \begin{bmatrix} C^*C & C^*D\\D^*C & D^*D - \lambda I\end{bmatrix} 
+ \begin{bmatrix} A^*PA-P & A^*PB\\B^*PA & B^*PB\end{bmatrix}
\right| \nonumber \\
&\leq \begin{bmatrix} 0 & 0\\ 0 & t\mathbf{1}\mathbf{1}^T\end{bmatrix}, \label{eq:complicated-dual}
\end{align}
where the inequality $\leq$ is a component-wise inequality. In addition, $d(P, \lambda, t) = +\infty$ if $(Q, P, \lambda, t)$ does not satisfy \eqref{eq:complicated-dual}.
By defining $Y=Y^*$ to be a right bottom block of \eqref{eq:complicated-dual}, we obtain the following dual program of \eqref{eq:sparse-hinfinity-sdp}.

\begin{equation}
\begin{aligned}
& \underset{P, Y, \lambda, t} {\text{minimize}}
& & \lambda + k\cdot t\\
&\text{subject to}
&&
\begin{bmatrix}
A^*PA - P & A^*PB\\
 B^*PA & B^*PB - \lambda I - Y
 \end{bmatrix}
\\
&&& +\begin{bmatrix}
C^*C & C^*D\\
D^*C & D^*D
\end{bmatrix}
\preceq 0\\
%\begin{bmatrix}
%A & B\\
%C & D
%\end{bmatrix}^*
%\begin{bmatrix}
%P & 0\\
%0 & I
%\end{bmatrix}
%\begin{bmatrix}
%A & B\\
%C & D
%\end{bmatrix}\\
%&&&
%-
% \begin{bmatrix}
%P & 0\\
%0 & \lambda I + Y
%\end{bmatrix}
%\preceq 0\\
&&&
|Y| \leq t\mathbf{1}\mathbf{1}^T\\
&&& P=P^*, Y=Y^*, t \geq 0, \lambda \geq 0.
\end{aligned}
\label{eq:dual-sparse-hinfinity-sdp}
\end{equation}

Notice that if we set $t = 0$, then we recover the SDP derived from the KYP lemma which computes the $\Hinf$ norm of the system.
It is clear that $t=0$ is a suboptimal solution of \eqref{eq:dual-sparse-hinfinity-sdp}, and therefore we can easily see that the $\mathcal{H}_{\infty}$ norm is an upper bound of \eqref{eq:dual-sparse-hinfinity-sdp} which is consistent with the definition of $k$-sparse $\mathcal{H}_{\infty}$ norm.

In addition, it can be easily checked that the dual problem \eqref{eq:dual-sparse-hinfinity-sdp} is strictly feasible when $A$ is stable by setting $Y = 0$, $t=1$ and sufficiently large $\lambda$.
With sufficiently large $\lambda$, only left upper block of the LMI constraint becomes relevant, and if $A$ is stable, we can find $P \succ 0$ such that $A^*PA - P + C^*C \prec 0$, and this gives a strictly dual feasible point.
Therefore duality gap is zero.

%[INTERPRETATION OF DUAL?]
Another observation is that if we assume $(A,B,C,D)$ are real matrices, then similar argument as in \cite{you2013lagrangian} shows that all matrices in \eqref{eq:dual-sparse-hinfinity-sdp} can be taken as real matrices.
In this case, the absolute value constraint becomes $-t\mathbf{1}\mathbf{1}^T \leq Y \leq t\mathbf{1}\mathbf{1}^T$, a familiar linear constraint used to impose an element wise $\ell_1$ norm bound on a matrix.

A similar derivation for a continuous time system, \eqref{eq:cont-sparse-hinfinity-sdp}, gives us
\begin{equation}
\begin{aligned}
& \underset{P, Y, \lambda, t} {\text{minimize}}
& & \lambda + k\cdot t\\
&\text{subject to}
&& 
\begin{bmatrix}
A^*P + PA & PB\\
B^*P & -\lambda I - Y
\end{bmatrix}\\
&&& + 
\begin{bmatrix} C^*C & C^*D\\D^*C & D^*D \end{bmatrix} \preceq 0\\
&&&
|Y| \leq t\mathbf{1}\mathbf{1}^T\\
&&& P=P^*, Y=Y^*, t \geq 0, \lambda \geq 0,
\end{aligned}
\label{eq:dual-sparse-c-hinfinity-sdp}
\end{equation}
although we omit the detailed derivation.

\subsection{$k$-sparse $\hinf$ synthesis}
%Based on the dual problem of our SDP relaxation \eqref{eq:dual-sparse-hinfinity-sdp}, using SDP thanks to its KYP like structure.
Here we follow the LMI approach to $\hinf$ controller synthesis from \cite{gahinet1994linear} in order to develop an SDP which designs a controller that minimizes the proposed sdp relaxation of $k$-sparse $\hinf$ measure.
To begin with, consider the dynamical system
\BEAS
x_{k+1} &=& Ax_{k} + B_1 w_{k} + B_2 u_{k}\\
z_{k} &=& C_1 x_{k} + D_{11} w_{k} + D_{12} u_k\\
y_k &=& C_2 x_k + D_{21} w_k ,
\EEAS
with dynamic controller
\BEAS
\zeta_{k+1} &=& A_K \zeta_k + B_K y_k\\
u_k &=& C_K \zeta_k + D_K y_k.
\EEAS
The synthesis goal is to design $A_K, B_K, C_K, D_K$ which minimizes our \emph{SDP relaxation} of the $k$-sparse $\hinf$ norm,  \eqref{eq:sparse-hinfinity-sdp}, of the closed loop system:
\BEAS
%x
\begin{bmatrix}
x_{k+1}\\
\zeta_{k+1}
\end{bmatrix}
&=&
\underbrace{
\begin{bmatrix}
A+B_2D_KC_2 & B_2 C_K\\
B_KC_2 & A_K
\end{bmatrix}}_{A_{cl}}
\begin{bmatrix}
x_{k}\\
\zeta_{k}
\end{bmatrix}\\
&&
+
\underbrace{
\begin{bmatrix}
B_1 + B_2D_KD_{21}\\
B_KD_{21}
\end{bmatrix}}_{B_{cl}}
w_k\\
%z
z_k &=& 
\underbrace{
\begin{bmatrix}
C_1 + D_{12}D_KC_2 & D_{12}C_K
\end{bmatrix}}_{C_{cl}}
\begin{bmatrix}
x_{k}\\
\zeta_{k}
\end{bmatrix}\\
&&+
\underbrace{
(D_{11} + D_{12}D_K D_{21})
}_{D_{cl}}
w_k
\EEAS
with stable $A_{cl}$, \ie, $\rho(A_{cl}) < 1$.
From this stability requirement, together with Lyapunov stability theorem, we can assume that $P \succ 0$ in the LMI in \eqref{eq:dual-sparse-hinfinity-sdp}.
In addition, we assume that the pair $(A,B_2, C_2)$ are {\it{stabilizable}}, \ie, there exists at least one controller $(A_K, B_K, C_K, D_K)$ such that $\rho(A_{cl}) < 1$, otherwise there exists no stabilizing controller so the controller design problem has no sense.

Using $(A_{cl}, B_{cl}, C_{cl}, D_{cl})$, the matrix inequality constraint in the dual of $k$-sparse $\hinf$ analysis, \eqref{eq:dual-sparse-hinfinity-sdp}, is given by
\begin{equation}
\begin{aligned}
\begin{bmatrix}
A_{cl}^*P_{cl}A_{cl} - P_{cl} & A_{cl}^*P_{cl}B_{cl}\\
 B_{cl}^*P_{cl}A_{cl} & B_{cl}^*P_{cl}B_{cl} - \lambda I - Y
\end{bmatrix}\\
+
\begin{bmatrix}
C_{cl}^*C_{cl} & C_{cl}^*D_{cl}\\
D_{cl}^*C_{cl} & D_{cl}^*D_{cl}
\end{bmatrix}
\prec 0
\end{aligned}
\label{eq:closed-loop-LMI}
\end{equation}
Here we have change the non-strict matrix inequality of optimization prpoblem \eqref{eq:dual-sparse-hinfinity-sdp} to a strict inequality.
Since $A_{cl}$ is required to be stable, a strictly feasible solution exists for a given $(A_{cl}, B_{cl}, C_{cl}, D_{cl})$, 
and therefore this variation does not change the optimal value of the following synthesis problem.
\begin{equation}
\begin{aligned}
& \underset{A_K, B_K, C_K, D_K, P_{cl}, Y, \lambda, t} {\text{minimize}}
& & \lambda + k\cdot t\\
&\text{subject to}
&&
\eqref{eq:closed-loop-LMI}, |Y| \leq t\mathbf{1}\mathbf{1}^{\top}, P_{cl} \succ 0\\
&&&Y=Y^*, t \geq 0, \lambda \geq 0.
\end{aligned}
\label{eq:synthesis}
\end{equation}
Notice that the above problem is not convex because the condition  \eqref{eq:closed-loop-LMI} is a bilinear matrix inequality.
However because of its similarity to the standard $\hinf$ synthesis problem, optimziation problem \eqref{eq:synthesis} can be solved using semidefinite programing so long as the order of the controller is taken to be greater than or equal to that of the plant.  To that end, let $T = \lambda I + Y$, then since $T - D_{cl}^*D_{cl} \succ 0$, $T$ is positive definite.
With left and right multiplication of $\begin{bmatrix} I & 0\\0 & T^{-1/2}\end{bmatrix}$ to \eqref{eq:closed-loop-LMI},
we obtain the following equivalent condition of \eqref{eq:closed-loop-LMI}:
\begin{equation}
\begin{aligned}
\begin{bmatrix}
A_{cl}^*P_{cl}A_{cl} - P_{cl} & A_{cl}^*P_{cl}B_{cl}T^{-1/2}\\
T^{-1/2}B_{cl}^*P_{cl}A_{cl} & T^{-1/2}B_{cl}^*P_{cl}B_{cl}T^{-1/2} - I
\end{bmatrix}\\
+
\begin{bmatrix}
C_{cl}^*C_{cl} & C_{cl}^*D_{cl} T^{-1/2}\\
T^{-1/2}D_{cl}^*C_{cl} & T^{-1/2}D_{cl}^*D_{cl}T^{-1/2}
\end{bmatrix}
\prec 0
\end{aligned}
\label{eq:syn-lmi}
\end{equation}
We can immediately see that the above condition is equivalent to the $\hinf$ norm of the {\it{scaled}} closed loop system, $(A_{cl}, B_{cl}T^{-1/2}, C_{cl}, D_{cl}T^{-1/2})$, is less than $1$
, \ie, $\|(C_{cl} (zI - A_{cl})B_{cl} + D_{cl})T^{-1/2}\|_{\infty} < 1$.
Therefore with this representation of \eqref{eq:closed-loop-LMI},
a simple modification of the classic LMI method for $\Hinf$ synthesis 
\cite{gahinet1994linear} yields the following two SDPs which solves \eqref{eq:synthesis} when the controller order is taken to be greater than equal to the plant order.

\begin{equation}
\begin{aligned}
& \underset{P, Q, Y, \lambda, t} {\text{minimize}}
& & \lambda + k\cdot t\\
&\text{subject to}
&&
\begin{bmatrix} P & I_n \\ I_n & Q \end{bmatrix} \succ 0\\
&&& \hspace{-2cm} 
\Pi_c^*
\begin{bmatrix}
AQA^*-A & AQC_1^* & {B}_1\\
C_1QA^* & C_1QC_1^*-I  & {D}_{11} \\
{B}_1^* &{D}_{11}^* & - \lambda I -Y
\end{bmatrix}
\Pi_c \prec 0\\
&&&\hspace{-2cm} 
\Pi_o^*
\begin{bmatrix}
A^*PA-A& A^*P{B}_1 & C_1^*\\
{B}_1^*PA & {B}_1^*P{B}_1- \lambda I - Y  & {D}_{11}^*\\
C_1 & {D}_{11} & -I
\end{bmatrix}
\Pi_o \prec 0\\
&&&|Y| \leq t\mathbf{1}\mathbf{1}^{\top}, Y=Y^*, t \geq 0,  \lambda \geq 0
\end{aligned}
\label{eq:PQ-synthesis}
\end{equation}
where
\begin{align*}
\Pi_c &= 
\begin{bmatrix}
N_c & 0\\
0 & I
\end{bmatrix}
&\Pi_o &= 
\begin{bmatrix}
N_o & 0\\
0 & I
\end{bmatrix}
\end{align*}
\begin{align*}
\range{N_c} &= \ker{\begin{bmatrix} B_2^* & D_{12}^*\end{bmatrix}}, &N_c^*N_c &= I\\
\range{{N}_o} &= \ker{\begin{bmatrix} C_2 & {D}_{21} \end{bmatrix}}, & {N}_o^*{N}_o &= I
\end{align*}

After obtaining $P, Q$ by solving the above SDP \eqref{eq:PQ-synthesis}, we construct $P_{cl}$ by $P_{cl} = \begin{bmatrix} P & P_2^*\\P_2 & I \end{bmatrix}$, where $P_2$ is given by $P-Q^{-1} = P_2P_2^*$.
By applying the Schur complement to  \eqref{eq:closed-loop-LMI}, we have
\begin{equation}
\begin{bmatrix}
P_{cl}^{-1} & 0 & A_{cl} & B_{cl}\\
0 & I & C_{cl} & D_{cl}\\
A_{cl}^* & B_{cl}^* & P_{cl} & 0\\
C_{cl}^* & D_{cl}^* & 0 & \lambda I + Y
\end{bmatrix}
\succ 0
\label{eq:controller-LMI}
\end{equation}
which is clearly an LMI for a fixed $P_{cl}$.
Finally, the following optimization returns the controller achieves the optimal value of \eqref{eq:PQ-synthesis}:
\begin{equation}
\begin{aligned}
& \underset{A_K, B_K, C_K, D_K, Y, \lambda, t} {\text{minimize}}
& & \lambda + k\cdot t\\
&\text{subject to}
&&
\eqref{eq:controller-LMI}, |Y| \leq t\mathbf{1}\mathbf{1}^{\top}, Y=Y^*\\
&&&t \geq 0, \lambda \geq 0.
\end{aligned}
\label{eq:controller-construction}
\end{equation}

In summary, two SDPs \eqref{eq:PQ-synthesis} and \eqref{eq:controller-construction} are needed to construct the controller.
The first step is to solve \eqref{eq:PQ-synthesis} to find $P, Q$ to construct $P_{cl}$, and then solve \eqref{eq:controller-construction} to find the controller $(A_K, B_K, C_K, D_K)$.

\section{Numerical examples}
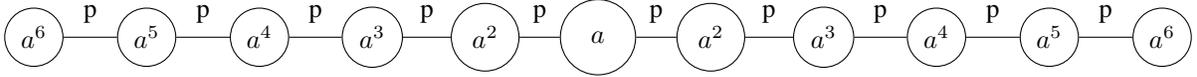
\begin{figure*}
\begin{center}
%graph 1
\begin{tikzpicture}
\node[draw,circle,minimum size=1cm](c) at (0,0) {$a$};
\node[draw,circle,minimum size=.9cm](r1) at (1.5,0) {$a^2$};
\node[draw,circle,minimum size=.9cm](l1) at (-1.5,0) {$a^2$};
\node[draw,circle,minimum size=.8cm](r2) at (3,0) {$a^3$};
\node[draw,circle,minimum size=.8cm](l2) at (-3,0) {$a^3$};
\node[draw,circle,minimum size=.7cm](r3) at (4.5,0) {$a^4$};
\node[draw,circle,minimum size=.7cm](l3) at (-4.5,0) {$a^4$};
\node[draw,circle,minimum size=.6cm](r4) at (6,0) {$a^5$};
\node[draw,circle,minimum size=.6cm](l4) at (-6,0) {$a^5$};
\node[draw,circle,minimum size=.5cm](r5) at (7.5,0) {$a^6$};
\node[draw,circle,minimum size=.5cm](l5) at (-7.5,0) {$a^6$};

\draw (c) -- (r1) node [midway, above=2pt] {p};
\draw (r1) -- (r2) node [midway, above=2pt] {p};
\draw (r2) -- (r3) node [midway, above=2pt] {p};
\draw (r3) -- (r4) node [midway, above=2pt] {p};
\draw (r4) -- (r5) node [midway, above=2pt] {p};

\draw (c) -- (l1) node [midway, above=2pt] {p};
\draw (l1) -- (l2) node [midway, above=2pt] {p};
\draw (l2) -- (l3) node [midway, above=2pt] {p};
\draw (l3) -- (l4) node [midway, above=2pt] {p};
\draw (l4) -- (l5) node [midway, above=2pt] {p};

\end{tikzpicture}
\caption{Homogeneous chain with $n = 5$}
\label{fig:chain}
\end{center}
\end{figure*}

In this section, we demonstrate the usefulness of our approach on various types of systems.  As will be seen, in many cases, the gap between our upper and lower bounds is very small, indicating that our relaxations are fairly tight.
For the optimization, we use CVX \cite{cvx} with SDPT3 \cite{toh1999sdpt3}.

\subsection{A linear chain}

Here we consider a linear chain with $2n+1$ nodes. Each subsystem has one internal state, and these states interact with adjacent states.
We assume that a disturbance can hit each state with unity gain, and the performance is the total sum of energy in each state.
This results in $B = I$, $C = I$, and $D = 0$, and $A \in \mathbb{R}^{2n+1 \times 2n+1}$ has the following form:
\BEAS
A_{ij}
:=
\begin{cases}
a^{|n+1-i|} & \text{if $i = j$}\\
p & \text{if $i = j+1$ or $i = j-1$}\\
0 & \text{otherwise}
\end{cases},
\EEAS
where we pick $a = 0.8$, and $p = 0.1$. 
See the Fig. \ref{fig:chain}.

Thanks to the system's symmetry, we can easily obtain the $k$-sparse $\mathcal{H}_{\infty}$ maximizing disturbance analytically.
The solution is to select the center node disturbance channel when $k = 1$, and as $k$ increases, including the right (or left) closest node from the center.
See Fig. \ref{fig:graph_chain} for the result. 
Here the semi-definite relaxation provides an upper bound and our rounding scheme provides a lower bound.
Due to its symmetry, the SDP relaxation has a hard time to find the actual solution, but interestingly enough, our rounding scheme returns the true optimal value.
We also compute the $\hinf$ norm of the system, and the ratio between $1$-sparse norm and $\hinf$ norm is around $0.85$.

\begin{figure}
\begin{center}
\includegraphics[width=0.45\textwidth]{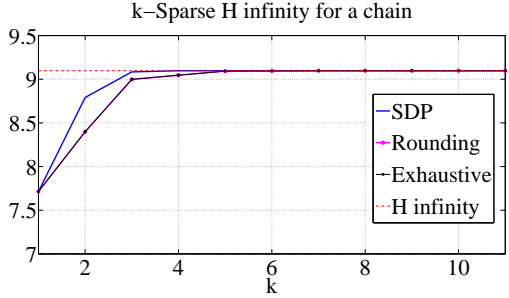}
\caption{SDP, rounding, and true value of $k$-sparse $\hinf$ analysis for a linear chain}
\label{fig:graph_chain}
\end{center}
\end{figure}

\subsection{Random dynamical system}
For this example, we construct an Erd\H{o}s-Renyi graph with $(n,p)$.
The weight of each edge is drawn from the standard Normal distribution to construct $A$.
Here, $B = 0.1 I_n$, $C = I_n$, and $D = 0$.

After obtaining values of $20$ samples, we plot the mean of the upper bound and the lower bound.
See Fig. \ref{fig:random_lti}.
We also perform exhaustive searches up to $k = 5$ to find the true optimal value.
In this case, we can see that the exact solution has matched with our rounding scheme.

\begin{figure}
\begin{center}
\includegraphics[width=0.45\textwidth]{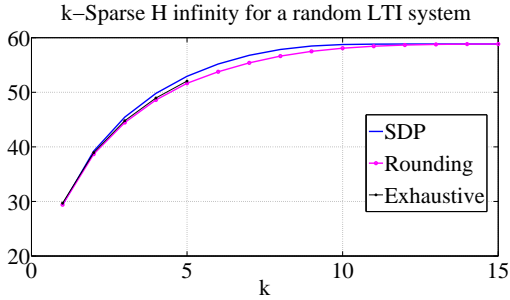}
\caption{Mean values of SDP, rounding, and true value of $k$-sparse $\hinf$ analysis for random LTI systems}
\label{fig:random_lti}
\end{center}
\end{figure}

\subsection{Synchronization network}
To construct the example of a synchronization network, we choose the Petersen graph for the graph topology.
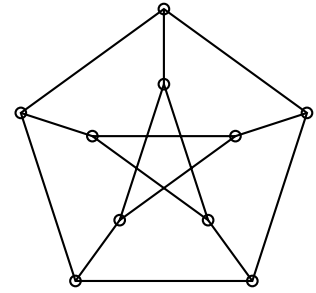
\begin{figure}
\begin{center}
\begin{tikzpicture}[style=thick]
\draw (18:2cm) -- (90:2cm) -- (162:2cm) -- (234:2cm) --
(306:2cm) -- cycle;
\draw (18:1cm) -- (162:1cm) -- (306:1cm) -- (90:1cm) --
(234:1cm) -- cycle;
\foreach \x in {18,90,162,234,306}{
\draw (\x:1cm) -- (\x:2cm);
\draw (\x:2cm) circle (2pt);
\draw (\x:1cm) circle (2pt);
}
\end{tikzpicture}
\caption{The Petersen graph}
\label{fig:petersen}
\end{center}
\end{figure}

Based on this topology, we generate two synchronization networks. The first one is based on the maximum degree rule, and the second one is based on the fastest protocol synthesis method via semidefinite programming \cite{Xiao:2004gg}.
See Fig. \ref{fig:petersen_max} and \ref{fig:petersen_fdla} for the result.

\begin{figure}
\begin{center}
\includegraphics[width=0.45\textwidth]{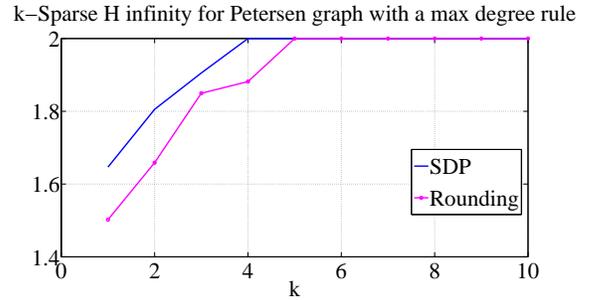}
\caption{SDP, and rounding value of $k$-sparse $\hinf$ analysis for the Petersen synchronization network with the maximum degree rule}
\label{fig:petersen_max}
\end{center}
\end{figure}

\begin{figure}
\begin{center}
\includegraphics[width=0.45\textwidth]{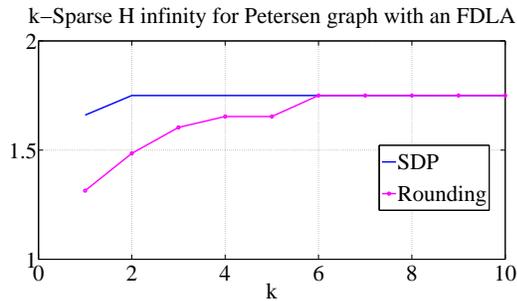}
\caption{SDP, and rounding value of $k$-sparse $\hinf$ analysis for the Petersen synchronization network with the maximum spectral gap}
\label{fig:petersen_fdla}
\end{center}
\end{figure}

\subsection{$k$-sparse $\hinf$ synthesis}
To illustrate the effectiveness of our synthesis approach, we apply our method to the following system:
\begin{align*}
A &= 
\begin{bmatrix}
0.5 & 0.2 & 0\\
0.2 & 0.5 & 0.2\\
0& 0.2 & 0.5
\end{bmatrix}
, ~~B_1=\begin{bmatrix} I_3 & 0_{3 \times 3} \end{bmatrix}
, ~~B_2= I_3\\
C_1 &= \begin{bmatrix} I_3\\0_{3\times 3} \end{bmatrix}
, ~~D_{11} = 0_{6 \times 6}
, ~~D_{12} = \begin{bmatrix} 0_{3\times 3}\\ I_3\end{bmatrix}\\
C_2 &= I_3 
, ~~D_{21} = \begin{bmatrix} 0_{3 \times 3} & I_3\end{bmatrix}
, ~~D_{22}= 0_{3\times 3}.
\end{align*}
Based on our approach, we obtain the controller that minimizes the SDP relaxation of the $k$-sparse $\hinf$ norm using convex optimization.
Then we compute the actual $k$-sparse $\hinf$ norm via exhaustive search to compare the results, see the Table \ref{table:controller} for the result.

%\begin{table}[t!]
%\centering
%\begin{tabular}{c|cccccc}
%$k$-sparse& \multicolumn{5}{c} {$k$-sparse $\hinf$ norm}\\
%controller & 1 & 2 & 3 & 4 & 5 & $\hinf$\\
%\hline
%1 &     1.1826  &  \color{gray}1.3939  &  \color{gray}1.5078   &\color{gray} 1.6171   &\color{gray} 1.6665 &  \color{gray} 1.7152\\
%2 &     \color{gray} 1.2289  &  1.3340  &  \color{gray}1.4116  & \color{gray} 1.4655  &\color{gray}  1.4984   & \color{gray}1.5258\\
%3 &     \color{gray}1.2509  &  \color{gray}1.3539 &   1.4053  &  1.4574  & \color{gray} 1.4892  &  \color{gray}1.5159\\
%4 &     \color{gray}1.2628 &  \color{gray} 1.3665  &  \color{gray}1.4194 &  \color{gray} 1.4772 &   \color{gray}1.5078 &  \color{gray} 1.5337\\
%5 &     \color{gray}1.2748  &\color{gray}  1.3821  & \color{gray} 1.4330   &\color{gray} 1.4888 &  \color{gray} 1.5194   &\color{gray} 1.5454\\
%$\hinf$  &   \color{gray}  1.3832   & \color{gray}1.4172 &   \color{gray}1.4389  &\color{gray} \color{gray} 1.4791 & 1.4882  &  1.5050
%\end{tabular}
%\caption{$k$-sparse $\hinf$ norms of the controller obtained from our method. The minimum values in the rows are emphasized. }
%\label{table:controller}
%\end{table}
\begin{table}[t!]
\centering
\begin{tabular}{c|cccc}
$k$-sparse& \multicolumn{4}{c} {$k$-sparse $\hinf$ norm}\\
controller & 1 & 2 & 3  & $\hinf$\\
\hline
1 &     1.1826  &  \color{gray}1.3939  &  \color{gray}1.5078   &  \color{gray} 1.7152\\
2 &     \color{gray} 1.2289  &  1.3340  &  \color{gray}1.4116     & \color{gray}1.5258\\
3 &     \color{gray}1.2509  &  \color{gray}1.3539 &   1.4053   &  \color{gray}1.5159\\
$\hinf$  &   \color{gray}  1.3832   & \color{gray}1.4172 &   \color{gray}1.4389   &  1.5050
\end{tabular}
\caption{$k$-sparse $\hinf$ norms of the controller obtained from our method. The minimum values in the rows are emphasized. }
\label{table:controller}
\end{table}

Since our synthesis method is based on the SDP relaxation of the $k$-sparse $\hinf$ norm, the resulting controller may not be the true optimal controller.
However, as we can see, the controllers computed with respect to relaxations of the $k$-sparse $\Hinf$ norm exhibit better performance with respect to $k$ disturbances than the general $\Hinf$ optimal controller.  In particular, if only $k$ disturbances are allowed to coordinate their attack, then we see that if a controller is designed to mitigate the worst case effect of a larger number of disturbances, this can in fact lead to a degradation in the closed loop $k$-spare $\Hinf$ norm of the system.

\section{Conclusion}
Motivated by robustness properties of large-scale systems, we defined the $k$-sparse $\hinf$ norm and the $k$-sparse minimal gain of a system, and argued that in such large-scale settings, traditional $\Hinf$ analysis may be overly conservative.  As computing these objects involves solving a combinatorial optimization problem, we developed semidefinite programming relaxations to these combinatorial optimization problems, akin to those used by \cite{d2007direct} in the context of sPCA, leading to upper (lower) bounds on the $k$-sparse $\hinf$ norm (minimal gain) of a system.  We also developed a simple rounding heuristic to provide corresponding lower (upper) bounds.  Via numerical simulation, we show that the gap between the upper bound and the lower bounds obtained is usually quite small, and that perhaps surprisingly, the candidate disturbance that achieves the lower bound is indeed the true worst-case sparse disturbance.  We also developed a centralized controller synthesis method based on the KYP-like dual of our semidefinite relaxation, and confirmed its effectiveness through a numerical example.  In future work, we aim to extend our synthesis methods to nested distributed systems \cite{scherer2013structured} (which also admit solutions based on KYP-like semidefinite programs), and to develop specialized solvers that exploit structure in the state-space parameters such that our methods can scale to systems with hundreds to thousands of states.

%\subfile{../chapters/conclusion}

\bibliographystyle{IEEEtran}
\bibliography{IEEEabrv,sparse_h,consensus}

% Generated by IEEEtran.bst, version: 1.12 (2007/01/11)
\begin{thebibliography}{10}
\providecommand{\url}[1]{#1}
\csname url@samestyle\endcsname
\providecommand{\newblock}{\relax}
\providecommand{\bibinfo}[2]{#2}
\providecommand{\BIBentrySTDinterwordspacing}{\spaceskip=0pt\relax}
\providecommand{\BIBentryALTinterwordstretchfactor}{4}
\providecommand{\BIBentryALTinterwordspacing}{\spaceskip=\fontdimen2\font plus
\BIBentryALTinterwordstretchfactor\fontdimen3\font minus
  \fontdimen4\font\relax}
\providecommand{\BIBforeignlanguage}[2]{{%
\expandafter\ifx\csname l@#1\endcsname\relax
\typeout{** WARNING: IEEEtran.bst: No hyphenation pattern has been}%
\typeout{** loaded for the language `#1'. Using the pattern for}%
\typeout{** the default language instead.}%
\else
\language=\csname l@#1\endcsname
\fi
#2}}
\providecommand{\BIBdecl}{\relax}
\BIBdecl

\bibitem{MP14_chordal}
R.~Mason and A.~Papachristodoulou, ``Chordal sparsity, decomposing sdps and the
  lyapunov equation,'' in \emph{American Control Conference (ACC), 2014}, June
  2014, pp. 531--537.

\bibitem{MLP14_admm}
C.~Meissen, L.~Lessard, and A.~Packard, ``Performance certification of
  interconnected systems using decomposition techniques,'' in \emph{American
  Control Conference (ACC), 2014}, June 2014, pp. 5030--5036.

\bibitem{RL06_QI}
M.~Rotkowitz and S.~Lall, ``A characterization of convex problems in
  decentralized control,'' \emph{Automatic Control, IEEE Transactions on},
  vol.~51, no.~2, pp. 274--286, Feb 2006.

\bibitem{matni2014regularization}
N.~Matni and V.~Chandrasekaran, ``Regularization for design,'' \emph{arXiv
  preprint arXiv:1404.1972}, 2014.

\bibitem{Dullerud:2010tc}
G.~E. Dullerud and F.~Paganini, \emph{{A course in robust control theory: A
  convex approach}}.\hskip 1em plus 0.5em minus 0.4em\relax Springer New York,
  2000.

\bibitem{candes2008restricted}
E.~J. Candes, ``The restricted isometry property and its implications for
  compressed sensing,'' \emph{Comptes Rendus Mathematique}, vol. 346, no.~9,
  pp. 589--592, 2008.

\bibitem{d2007direct}
A.~d'Aspremont, L.~El~Ghaoui, M.~I. Jordan, and G.~R. Lanckriet, ``A direct
  formulation for sparse {PCA} using semidefinite programming,'' \emph{SIAM
  review}, vol.~49, no.~3, pp. 434--448, 2007.

\bibitem{you2014h}
S.~You and A.~Gattami, ``H-infinity analysis revisited,'' \emph{arXiv preprint
  arXiv:1412.6160}, 2014.

\bibitem{doyle1992feedback}
J.~C. Doyle, B.~A. Francis, and A.~Tannenbaum, \emph{{Feedback control
  theory}}.\hskip 1em plus 0.5em minus 0.4em\relax Macmillan Publishing Company
  New York, 1992.

\bibitem{you2013lagrangian}
S.~You and J.~C. Doyle, ``{A Lagrangian dual approach to the Generalized KYP
  lemma},'' in \emph{CDC}, 2013, pp. 2447--2452.

\bibitem{OlfatiSaber:cm}
R.~Olfati-Saber, J.~A. Fax, and R.~M. Murray, ``Consensus and cooperation in
  networked multi-agent systems,'' \emph{Proceedings of the IEEE}, vol.~95,
  no.~1, pp. 215--233, 2007.

\bibitem{Boyd:bu}
S.~Boyd, A.~Ghosh, B.~Prabhakar, and D.~Shah, ``Randomized gossip algorithms,''
  \emph{Information Theory, IEEE Transactions on}, vol.~52, no.~6, pp.
  2508--2530, 2006.

\bibitem{jadbabaie2003coordination}
A.~Jadbabaie, J.~Lin, and A.~S. Morse, ``Coordination of groups of mobile
  autonomous agents using nearest neighbor rules,'' \emph{Automatic Control,
  IEEE Transactions on}, vol.~48, no.~6, pp. 988--1001, 2003.

\bibitem{Xiao:2004gg}
L.~Xiao and S.~Boyd, ``Fast linear iterations for distributed averaging,''
  \emph{Systems \& Control Letters}, vol.~53, no.~1, pp. 65--78, 2004.

\bibitem{bickel2009simultaneous}
P.~J. Bickel, Y.~Ritov, and A.~B. Tsybakov, ``Simultaneous analysis of lasso
  and dantzig selector,'' \emph{The Annals of Statistics}, pp. 1705--1732,
  2009.

\bibitem{chandrasekaran2012convex}
V.~Chandrasekaran, B.~Recht, P.~A. Parrilo, and A.~S. Willsky, ``The convex
  geometry of linear inverse problems,'' \emph{Foundations of Computational
  mathematics}, vol.~12, no.~6, pp. 805--849, 2012.

\bibitem{gattami2013simple}
A.~Gattami and B.~Bamieh, ``A simple approach to {$H_\infty$} analysis,'' in
  \emph{Decision and Control (CDC), 2013 IEEE 52nd Annual Conference on}.\hskip
  1em plus 0.5em minus 0.4em\relax IEEE, 2013, pp. 2424--2428.

\bibitem{tibshirani1996regression}
R.~Tibshirani, ``Regression shrinkage and selection via the lasso,''
  \emph{Journal of the Royal Statistical Society. Series B (Methodological)},
  pp. 267--288, 1996.

\bibitem{mirsky1975trace}
L.~Mirsky, ``A trace inequality of {John von Neumann},'' \emph{Monatshefte
  f{\"u}r Mathematik}, vol.~79, no.~4, pp. 303--306, 1975.

\bibitem{boyd2004convex}
S.~P. Boyd and L.~Vandenberghe, \emph{Convex optimization}.\hskip 1em plus
  0.5em minus 0.4em\relax Cambridge university press, 2004.

\bibitem{kim2007interior}
S.-J. Kim, K.~Koh, M.~Lustig, S.~Boyd, and D.~Gorinevsky, ``An interior-point
  method for large-scale {$l_1$}-regularized least squares,'' \emph{Selected
  Topics in Signal Processing, IEEE Journal of}, vol.~1, no.~4, pp. 606--617,
  2007.

\bibitem{gahinet1994linear}
P.~Gahinet and P.~Apkarian, ``A linear matrix inequality approach to
  $h_{\infty}$ control,'' \emph{International journal of robust and nonlinear
  control}, vol.~4, no.~4, pp. 421--448, 1994.

\bibitem{cvx}
M.~Grant and S.~Boyd, ``{CVX}: Matlab software for disciplined convex
  programming, version 2.0 beta,'' \url{http://cvxr.com/cvx}, Sep. 2012.

\bibitem{toh1999sdpt3}
K.-C. Toh, M.~J. Todd, and R.~H. T{\"u}t{\"u}nc{\"u}, ``{SDPT3Ña MATLAB
  software package for semidefinite programming, version 1.3},''
  \emph{Optimization methods and software}, vol.~11, no. 1-4, pp. 545--581,
  1999.

\bibitem{scherer2013structured}
C.~W. Scherer, ``Structured $\mathcal{H}_\infty$-optimal control for nested
  interconnections: A state-space solution,'' \emph{Systems \& Control
  Letters}, vol.~62, no.~12, pp. 1105--1113, 2013.

\end{thebibliography}

\end{document}